\documentclass[conference, a4paper,10pt]{IEEEtran}
\usepackage[T1]{fontenc}%
\pdfoutput=1
\IEEEoverridecommandlockouts

\usepackage[font=footnotesize,caption=false]{subfig} %
\usepackage{tikz}
\usepackage{cite}
\usepackage{amsmath,amssymb,amsfonts}
\usepackage{tikz}
\usepackage{tikz}

\usepackage{pgfplots}
\pgfplotsset{width=14cm,compat=1.9}
\usepackage[noend,ruled,algo2e]{algorithm2e}
\usetikzlibrary{arrows}
\usetikzlibrary{arrows.meta}
\usetikzlibrary{automata, positioning}
\usepackage{algorithmic}
\usepackage{graphicx}
\usepackage{textcomp}
\usepackage{xcolor}
\usepackage{booktabs}
\usepackage[short]{optidef}
\usepackage{mathtools}
\usepackage{glossaries}
\usepackage{subfig}
\usepackage{algorithm}
\usepackage{bbm}
\usepackage{multirow}

\usepackage{mathtools,amsthm}
\newtheorem{theorem}{Theorem}

\newtheorem{lemma}{Lemma}

\DeclareMathOperator*{\argmin}{arg\,min}

\def\BibTeX{{\rm B\kern-.05em{\sc i\kern-.025em b}\kern-.08em
    T\kern-.1667em\lower.7ex\hbox{E}\kern-.125emX}}
    
\def \fwidth{0.9\columnwidth}
\def \fheight {0.45\columnwidth}
\def \sfwidth{0.9\columnwidth}
\def \sfheight {0.45\columnwidth}

\definecolor{color3}{HTML}{FFD700}
\definecolor{color2}{HTML}{EA5F94}
\definecolor{color1}{HTML}{9D02D7}
\definecolor{color0}{HTML}{0000FF}

\title{A Decentralized Policy for Minimization of Age of Incorrect Information in Slotted ALOHA Systems}

\author{%
\IEEEauthorblockN{Anupam Nayak\IEEEauthorrefmark{1}, Anders E. Kalør\IEEEauthorrefmark{2}, Federico Chiariotti\IEEEauthorrefmark{2}\IEEEauthorrefmark{3}, and Petar Popovski\IEEEauthorrefmark{2}}%
\IEEEauthorblockA{\IEEEauthorrefmark{1}Department of Electrical Engineering, IIT Bombay, India (anupam@ee.iitb.ac.in)}
\IEEEauthorblockA{\IEEEauthorrefmark{2}Department of Electronic Systems, Aalborg University, Denmark
(\{aek,fchi,petarp\}@es.aau.dk)}
\IEEEauthorblockA{\IEEEauthorrefmark{3}Department of Information Engineering, University of Padova, Italy (chiariot@dei.unipd.it)}}

\begin{document}

\maketitle

\begin{abstract}
The Age of Incorrect Information (AoII) is a metric that can combine the freshness of the information available to a gateway in an Internet of Things (IoT) network with the accuracy of that information. As such, minimizing the AoII can allow the operators of IoT systems to have a more precise and up-to-date picture of the environment in which the sensors are deployed. However, most IoT systems do not allow for centralized scheduling or explicit coordination, as sensors need to be extremely simple and consume as little power as possible. Finding a decentralized policy to minimize the AoII can be extremely challenging in this setting. This paper presents a heuristic to optimize AoII for a slotted ALOHA system, starting from a threshold-based policy and using dual methods to converge to a better solution. This method can significantly outperform state-independent policies, finding an efficient balance between frequent updates and a low number of packet collisions.
\end{abstract}

\section{Introduction}
Transmission policies that guarantee a timely access to wireless sensor data represent an important component for Internet of Things (IoT) sensor networks, which can be applied to environmental monitoring, industrial automation, and several other critical systems. This demand for timely information has motivated the definition of Age of Information (AoI)~\cite{kaul12}, which has been analyzed by a wide range of research studies over the past years, manifesting the relevance of timeliness in sensor networks~\cite{kosta17,yates21}. In addition to the AoI, a number of related metrics have been proposed, including the Value of Information (VoI)~\cite{ayan19}, which captures the error of the estimated sensor values at the sink, thereby allowing the policy to take into account the dynamics of the observed processes. The Age of Incorrect Information (AoII)~\cite{10.1109/TNET.2020.3005549} is a proposal that attempts to combine the two metrics, being defined as the product of the AoI and VoI.

Beyond the analysis of AoII in a variety of systems, there is a significant challenge of implementing policies and protocols that reduce AoII for low-power IoT sensor networks. A simple solution for minimizing AoI is to use a pull-based scheme, in which the receiver can poll sensors and schedule the one with the highest age. However, this solution has two issues: firstly, the receiver knows the AoI for each sensor, but not the VoI, so it can only optimize for the statistical value, and might schedule sensors that do not actually have valuable information even if their age is high. Secondly, polling can be an energy-intensive process, and many IoT communication technologies do not even support it. Indeed, in order for the gateway to be able to initiate the communication, sensors need to listen for requests at any time, which can quickly deplete their batteries.

The optimization of AoI and AoII is often performed using Markov models, which can optimize expected performance. In~\cite{KSA21}, the authors analyze Whittle index-based heuristics for minimizing AoI, achieving asymptotic optimality when the number of users is large. Whittle index policies can also be used for AoII, as in~\cite{9518209}, which considers a multi-sensor scenario with an infinite time horizon. Many recent works also involve the use of threshold based policies. \cite{Bjoshi} proposes threshold-type policies to minimize the AoII with estimation at the receiver in absence of updates for a single source observing an autoregressive Markov process. A policy which is a randomized combination of two deterministic threshold policies is shown to be optimal in \cite{YA21}, where the authors study problem of minimizing AoII with power constraints in the presence of an unreliable channel. \cite{delay} shows the optimality of a threshold based pre-emptive policy and achieves the minimum AoII under a uniform and bounded delay distribution.

However, most of these works either only deal with the first issue by considering statistical optimization, so that the system is still based on a centralized scheduling, or consider single-sensor systems, in which collisions are nonexistent. In real push-based communications, sensors need to be able to independently access the channel, while coordinating with each other in a distributed manner. This problem is significantly more challenging, as the information each sensor has is severely limited. An aggressive strategy can cause the network to become unstable, leading to extremely frequent collisions, while a too conservative one can starve the receiver, leading to an extremely high error rate. This decentralized setting has been studied in the context of the basic AoI metric. Threshold based policies for minimizing AoI in slotted ALOHA have been proposed in \cite{9162973,vavascan21}, where the sensors stay silent until they have a certain AoI, after which they transmit with a constant probability.

The objective of this work is to minimize AoII in a decentralized setting in which a group of IoT devices communicate to the sink over a shared slotted ALOHA random access channel. To this end, we formulate the AoII minimization task as a non-convex optimization problem, proposing a gradient-based algorithm to obtain an approximate solution. Our results show that the proposed algorithm manages to find policies that significantly outperform state-independent policies, and suggest that the optimal policy is a threshold-like function that depends both on the current AoI and the VoI.

The remainder of the paper is organized as follows. First, we present the system model in Sec.~\ref{sec:sysmodel}. We then analyze the AoII evolution in Sec.~\ref{sec:aoii_evolution} under an arbitrary policy, and derive an upper bound on the average AoII which we use in Sec.~\ref{sec:aoii_optim} to optimize the policy. Finally, we present numerical results in Sec.~\ref{sec:results} and conclude the paper in Sec.~\ref{sec:conclusion}.

\section{System Model}\label{sec:sysmodel}

The system we consider involves a remote BS that aims to collect observations from $N$ different sensors which observe independent, identically distributed discrete Markov processes. Specifically, at time $t=1,2,\ldots$ the $i$-th sensor observes $X_i(t) \in \mathbb{Z}$, a value governed by the random walk with the transition probability diagram depicted in Fig.~\ref{fig:2}. Note that $p_r+2p_t =1$, to ensure that they represent a valid probability space. Each sensor is responsible for communicating its observed value to the BS over a shared medium, whose access is regulated by slotted ALOHA: in every slot, any sensor may access the channel, and interference is destructive, so that if multiple sensors transmit in the same time slot, no packet is received. We assume that the BS acknowledges the successfully received packets, and will denote by $\hat{X}_i(t)$ the state of the Markov process from sensor $i$ that was most recently communicated to the BS. We do not assume that the nodes follow a specific retransmission strategy, but instead model retransmissions implicitly as part of the policy.

To jointly characterize the freshness and accuracy of the information the BS has, we define the AoII for sensor $i$ as
\begin{equation}
    A_i(t) = f_i(t)g_i(t),
\end{equation}
where $f_i(t)$ corresponds to the penalty for information freshness (age), while $g_i(t)$ accounts for a penalty that occurs due to the difference in the actual information. 
We will work with a linear time penalty given as $f_i(t) = t - U_i(t)$, where $U_i(t)$ corresponds to the last time instance when $X_i$ and $\Hat{X}_i$ were equal, i.e.,
\begin{equation}
   U_i(t) = \max \{t_i | t_i \leq t , X_i(t_i) = \Hat{X}_i(t_i) \}.
\end{equation}
For simplicity we define $\Hat{X}_i(0)=X_i(0)=0$.
Note that, according to this definition, $f_i(t)$ is \emph{not} the same as the time elapsed since last successful transmission by sensor $i$, but the time elapsed since the observed state was the same as the most recently transmitted state. Clearly, $f_i(t)$ is computable by sensor $i$, which observes the state transitions. The penalty term $g_i(t)$ is simply given by the difference between the two:
\begin{equation}
g_i(t) = |X_i(t)-\Hat{X}_i(t)|.
\end{equation}

The objective of this work is to devise a common transmission policy for all sensors, i.e, a mapping $\pi:\mathbb{Z}^2\rightarrow[0,1]$ that assigns a transmission probability for every pair $(f,g)$ of age and correctness penalties. Specifically, we seek a policy that minimizes the expected AoII averaged across the $N$ sensors:
\begin{equation}\label{eq:expected_aoii}
    \bar{A}=\mathbb{E}\left[\limsup_{T\to\infty}\frac{1}{TN}\sum_{t=1}^T\sum_{i=1}^N A_i(t)\middle\vert\pi\right],
\end{equation}
where the expectation is over the state evolution and the decisions of the transmission policy.

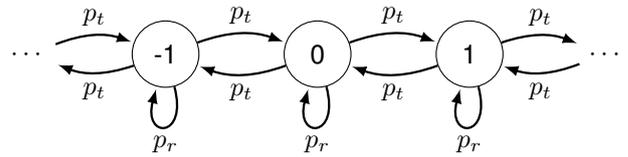
\begin{figure}
\centering
\begin{tikzpicture}[font=\sffamily]
   \node[state] at (0, 0) (S_m) {0};
   \node[state] at (2, 0) (S_1) {1};
   \node[state] at (-2, 0) (S_m2) {-1};
   \node[draw=none,  right=of S_1]   (m) {$ \cdots$};
  \node[draw=none,  left=of S_m2]   (k) {$ \cdots$};
    \draw[every loop, line width=0.3mm,
          >=latex,draw=black,
          fill=black]
      (m) edge[bend left=20, auto=left] node {$p_t$} (S_1)
      (S_1) edge[bend left=20, auto=left] node {$p_t$} (m)
      (k) edge[bend left=20, auto=left] node {$p_t$} (S_m2)
      (S_m2) edge[bend left=20, auto=left] node {$p_t$} (k)
      (S_m) edge[bend left=20, auto=left] node {$p_t$} (S_1)
      (S_1) edge[bend left=20, auto=left] node {$p_t$} (S_m)
(S_m2) edge[bend left=20, auto=left] node {$p_t$} (S_m)
(S_m) edge[bend left=20, auto=left] node {$p_t$} (S_m2)
      (S_1) edge[loop below, auto=left] node {$p_r$} (S_1)
      (S_m) edge[loop below, auto=left] node {$p_r$} (S_m)
(S_m2) edge[loop below, auto=left] node {$p_r$} (S_m);

\end{tikzpicture}
\caption{Process observed at each sensor.}
\label{fig:2}
\end{figure}

\section{AoII Evolution Analysis}\label{sec:aoii_evolution}
In this section, we derive the transition probabilities of the Markov chain describing the evolution of the AoII under a given policy $\pi$, which we will use in the next section to optimize the transmission policy. However, as the optimization problem is significantly complicated by the dependence/collisions among the sensors over the shared channel, at first we relax this constraint and assume that each user sees a fixed probability of success, computed as a steady state approximation. Consequently, we will focus on the evolution of $(f_i(t),g_i(t))$ for an arbitrary user $i$. Similar approximations have been used previously~\cite{LD12,BJK05} to analyse the performance of various random access protocols.

\subsection{Markov Model Formulation}\label{sec:markov_chain}
Since $f_i(t)$ and $g_i(t)$ are countably infinite, we truncate the Markov chain to $f_i(t)\leq F$, $g_i(t) \leq G$, setting fixed values $F$ and $G$, so that the mapping function becomes $\pi:\{0,\ldots,F\}\times\{0,\ldots,G\}\rightarrow[0,1]$. Assuming that $F$ and $G$ are sufficiently large, this assumption comes with a negligible impact on the optimal policy.
We define the truncated state as $f^{\tau}_i(t) = \min\{f_i(t),F\}$ and $g^{\tau}_i(t) = \min\{g_i(t),G\}$ so that the state space for each node is given by
\begin{equation}
    \mathcal{S} = \left\{(f,g)\;|\;f \leq F, \;g \leq \min(G,f) \right\}.
\end{equation}
Let $\pi(f^{\tau}_i(t),g^{\tau}_i(t))$ denote the probability that sensor $i$ transmits in state $(f^{\tau}_i(t),g^{\tau}_i(t))$. Denote further by $s_i(t)\in\mathcal{S}$ the tuple $(f^{\tau}_i(t),g^{\tau}_i(t))$, and by $q(s_i(t))$ denote the probability of successful transmission. Note that $\pi(s_i)\geq q(s_i)\,\forall s_i$, as not all transmissions are successful.

Any state for which $g^{\tau}_i(t)=0$ is collapsed into state $(0,0)$, due to the definition of the age penalty. We can then define the probability function $P(s,s')$, mapping the transition probabilities.
We also define the symbol $[x+y]_T=\min(T,x+y)$ to simplify the notation below.
 
 The simplest case is state $(0,0)$, in which we have:
 \begin{equation}
     P((0,0),s')=\begin{cases}
     2p_t(1-q((0,0))), &f'=g'=1;\\
     p_r+q((0,0)), & f'=g'=0,
     \end{cases}
 \end{equation}
 where $s'=(f',g')$. If we consider state $(f,1)$, we get:
  \begin{equation}
     P((f,1),s')\!=\!\begin{cases}
     p_t(1-q((f,1))), &f'\!=[f\!+\!1]_F,g'\!=2;\\
     p_r(1-q((f,1))), &f'\!=[f\!+\!1]_F,g'\!=1;\\
     p_t\!+\!(1\!-p_t)q((f,1)), &f'\!=g'\!=0,
     \end{cases}
 \end{equation}
 In all other cases, we have:
 \begin{equation}
     P(s,s')=\begin{cases}
     p_t(1-q(s)), &\begin{aligned} &f'=[f+1]_F,\\
     &g'\in\{[g+1]_G,g-1\};\end{aligned}\\
     p_r(1-q(s)), &f'=[f+1]_F,g'=g;\\
     q((s)), & f'=g'=0.
     \end{cases}
 \end{equation}

The chain on $\mathcal{S}$ defined by the transition probabilities above is a truncated version of the true AoII evolution.
Additionally, it is finite, aperiodic, and irreducible, i.e., every state can be reached from $(0,0)$ with positive probability and $(0,0)$ can be reached from every state with positive probability. Given these conditions, the chain has a stationary distribution $\phi(f,g)$.

\subsection{Probability of Successful Transmission}
The probability of successful transmission $q(s)$ used in the previous section depends on the policy $\pi$. Recall that we consider the case in which all sensors have the same policy, but act in an entirely independent fashion, and that we for simplicity assume that each user sees a fixed probability of success in any given state. Before deriving an expression for $q(s)$, we first present the following intermediate result.
\begin{lemma}\label{lemma:gminf}
For any reachable state $(f,g)$ with $f \neq F$, we have $g \leq f$.
\end{lemma}
\begin{proof}
In order for $g$ to have a certain value, there must have been at least $g$ transitions starting from value 0. These transitions would require at least $f$ steps, as the chain can only increase or decrease by 1 in each step.
\end{proof}

Assuming that all sensors are in steady state, except the sensor of interest, we have $q(f,g) = \pi(f,g)\ell^{N-1}$,
where $\ell$ corresponds to the probability that no other sensor is transmitting. Using Lemma~\ref{lemma:gminf}, the value of $\ell$ can be computed as:
\begin{equation}
    \ell = \sum_{f=0}^{F}\sum_{g=1}^{\min(f,G)}\phi(f,g)(1-\pi(f,g)).
\end{equation}

\subsection{Upper Bound on the AoII}
We conclude the section by presenting an upper bound on the AoII.
We can give the expected truncated AoII (which is a lower bound to the real AoII) by applying the ergodic theorem of Markov chains:
\begin{equation}
    \mathbb{E}[fg] \geq \sum_{f=0}^{F}\sum_{g=0}^{\min(f,G)}fg\phi(f,g).
\end{equation}
Note that this value depends on $\pi$, as the steady state distribution is a function of the transmission policy. As the truncated AoII is a lower bound to the real AoII, it is not suitable for reliability-oriented optimization, as using it leads to optimistic policies and does not return a reliable estimate of the AoII.
We can then derive the following result.
\begin{theorem}\label{theo:upper_bound}
For any policy $\pi:\mathcal{S}\to[0,1]$ and a stationary distribution $\phi$ for the Markov chain derived in Section~\ref{sec:markov_chain}, the average truncated AoII $\mathbb{E}[fg]$ is upper bounded by
\begin{equation}
\begin{aligned}
    J(\pi,\phi)= FG\phi(F,G) +\sum_{f=1}^{F-1}\left[\phi(f,G)f^2+\sum_{g=1}^{\mathclap{\min(f,G-1)}}fg\phi(f,g)\right]\\
    +\sum_{g=1}^{G-1}G\phi(F,g)\left(F+\frac{(1-\min_{g\in\{1,\ldots,G-1\}}q(F,g))}{\min_{g\in\{1,\ldots,G-1\}}q(F,g)}\right)\\
    + \left(F+G+\frac{2(1-q((F,G))}{q((F,G))}\right)\frac{1-q((F,G))}{q((F,G))}\phi(F,G).
\end{aligned}
\end{equation}
\end{theorem}
\begin{proof}
Let us consider the two subsets of $\mathcal{S}$ given by
\begin{align}
    \mathcal{S}_F=\left\{(F,g):0<g<G\right\};
    \mathcal{S}_G=\left\{(f,G):G\leq f<F\right\}; \nonumber
\end{align}
The only possible transitions out of these classes are to states $(0,0)$, for successful transmission, or $(F,G)$, if the other limit is reached. Transitions within each class are possible, as defined by $P$. The non-truncated AoII, $A_i(t)=f_i(t)g_i(t)$, is different from the truncated AoII only if the truncated chain is in one of the states belonging to $\mathcal{S}_f$, $\mathcal{S}_g$ or $(F,G)$.

The probability of exiting $\mathcal{S}_f$ through a successful transmission at any given time, denoted as $\omega_F$, is given by
\begin{equation}
    \omega_F=\frac{\sum_{s\in\mathcal{S}_f}\phi(s)q(s)}{\sum_{s\in\mathcal{S}_f}\phi(s)}.
\end{equation}
This is lower-bounded by $\min_{s\in\mathcal{S}_F}q(s)$, so that the geometric distribution with parameter $\min_{s\in\mathcal{S}_F}q(s)$ represents an upper bound to the time spent in $\mathcal{S}_F$.
The AoII after being in $\mathcal{S}_F$ for $i$ slots is upper-bounded by $(F+i)G$, as all states in $\mathcal{S}_F$ have an error smaller than $G$. We can adopt the same approach for state $(F,G)$, as the AoII after remaining in the state for $i$ steps is upper-bounded by $(F+i)(G+i)$. Finally, we consider class $\mathcal{S}_G$: thanks to Lemma~\ref{lemma:gminf}, the AoII in state $(f,G)$ is upper-bounded by $f^2$. We can then join the pieces to obtain the bound

\begin{multline}
\mathbb{E}[A_i(t)] 
\leq \sum_{f=1}^{F-1}\quad\quad\sum_{g=1}^{\mathclap{\min(G-1,f)}}\phi(f,g)fg\; + \; \underbrace{\sum_{f=1}^{F-1}\phi(f,G)f^2}_{\text{U.B. for states in } \mathcal{S}_G}\\ +\underbrace{\sum_{i=0}^{\infty}\left[ (\min_{s\in\mathcal{S}_F}q(s))(1-\min_{s\in\mathcal{S}_F}q(s))^i(F+i)G \sum_{s\in\mathcal{S}_F}\phi(s)\right] }_{\text{U.B. for states in } \mathcal{S}_F}\\
\underbrace{\sum_{i=0}^{\infty} \phi((F,G))(q((F,G)))(1-q((F,G)))^i(F+i)(G+i)}_{\text{U.B. for state }(F,G)}
\end{multline}
We can solve the second term in the sum as follows:
\begin{multline}
 \sum_{i=0}^{\infty}\left[ (\min_{s\in\mathcal{S}_F}q(s))^i(1-\min_{s\in\mathcal{S}_F}q(s))(F+i)G \sum_{s\in\mathcal{S}_F}\phi(s)\right]=\\
 G\sum_{s\in\mathcal{S}_F}\phi(s)\left(F+\frac{1-\min_{s\in\mathcal{S}_F}q(s)}{\min_{s\in\mathcal{S}_F}q(s)}\right)=\\
 \sum_{g=1}^{G-1}G\phi(F,g)\left(F+\frac{(1-\min_{s\in\mathcal{S}_F}q(s)}{\min_{s\in\mathcal{S}_F}q(s)}\right).
\end{multline}
Finally, the series giving the upper bound for state $(F,G)$ can be solved as below, omitting the term $\phi((F,G))$:
\begin{multline}
\sum_{i=0}^{\infty} q((F,G))(1-q((F,G)))^i(F+i)(G+i)=\\
\left[FG+\frac{(1-q((F,G))}{q((F,G))}\left(F+G+\frac{2(1-q((F,G))}{q((F,G))}\right)\right].
\end{multline}
If we sum the components, we obtain the value of $J(\pi,\phi)$.
\end{proof}
Theorem~\ref{theo:upper_bound} provides a closed-form upper bound on the system under the steady state assumption, which we can use to search for a policy $\pi^*$ that minimizes approximated AoII. Such a solution is expected to also have a small actual AoII.

\section{AoII optimization}\label{sec:aoii_optim}
\subsection{Problem Definition}
We can now define an optimization problem over the policy space. Our objective is to find the policy $\pi^*$ that minimizes the average expected AoII, and can be defined as
\begin{equation}
    \pi^*=\argmin_{\pi:\mathcal{S}\rightarrow[0,1]}
    \mathbb{E}\left[\limsup_{T\to\infty}\frac{1}{TN}\sum_{t=1}^T\sum_{i=1}^N f_i(t)g_i(t)\middle\vert\pi\right].
\end{equation}
As mentioned, we approximate the minimization using the bound $J(\pi,\phi)$ derived in Theorem~\ref{theo:upper_bound}. To this end, we define the following problem:
\begin{mini}
{(\pi,\phi)}{J(\pi,\phi)}
{}{}
\addConstraint{\phi P = P,\ \ \|\phi\|_1=1}
\addConstraint{0 \leq \phi(s) \leq 1\,\forall s\in\mathcal{S}}
\addConstraint{0 \leq \pi(s) \leq 1\,\forall s\in\mathcal{S}.}\label{eq:mini1}
\end{mini}
Here, the constraint $\phi P = P$ represents the equality conditions for the steady state distribution, where $P$ here is the transition probability matrix for the Markov chain derived in Section~\ref{sec:markov_chain}, which is a function of both $\pi$ and $\phi$. The other constraints ensure that $\phi$ and $\pi(s)$ are probability distributions.

Because of the constraints, \eqref{eq:mini1} is difficult to solve directly. Instead, we consider the Lagrangian relaxation to the above problem given as
\begin{mini}[2]
{(\pi,\phi)}{J(\pi,\phi)+k_1c_1 +k_2c_2+k_3c_3+k_4c_4,}
{}{}
\label{opt1}
\end{mini}
where the four penalty terms are defined as:
\begin{equation}
\begin{aligned}
c_1 &= (\|\phi P - \phi\|_2)^2;\\
c_2 &= \|\pi-1\|^2\mathbb{I}\{\pi>1\}+\|\pi\|^2\mathbb{I}\{\pi<0\};\\
c_3 &=(\|\phi\|_1- 1)^2;\\
c_4 &= (\|\phi-\Hat{1}\|)^2\mathbb{I}\{\phi>1\}+\|\phi\|^2\mathbb{I}\{\phi<0\},
\end{aligned}
\end{equation}
where $k_1,k_2,k_3,k_4$ are nonnegative Lagrange multipliers that penalize the violation of the constraints. As the Lagrange multiplier $k_i$ is increased, the constraint becomes tighter, i.e. the resulting policy chosen by the algorithm has lower $c_i$.

\subsection{Optimization Algorithm}
Solving the relaxed minimization problem on a computer involves rewriting the objective function as
\begin{mini}[2]
{(\pi,\phi)}{J(\pi,\phi)+\sum_{i=1}^4k_i\rho(c_i - \epsilon_i),}
{}{}
\label{opt2}
\end{mini}
where $\rho(x)=\max(0,x)$ is the  rectified linear unit (ReLU) function, $\epsilon_i$ is a tolerance level for constraint $i$ quantifying satisfactory performance on constraint $i$. Ideally, in problem (\ref{opt1}) one should increase $k_i$ in steps to $\infty$ which ensures that the solution obtained have $c_i \rightarrow 0$ asymptotically. However trying to make $k_i c_i$ arbitrarily small can result in very high values of $J(\pi,\phi)$, and is a procedure of high complexity.
 
 Using tolerances can help us avoid having to increase $k_i$ to $\infty$ trying to get arbitrarily small $c_i$ as anything below the tolerance level $\epsilon_i$ will fetch a zero penalty. Thus increasing $k_i$ will not change the solution. In our experiments we set $k_i$ to a constant large value in the initialization itself. However our experiments suggest that leaky ReLU, $\rho(x)=\max(\rho_a x,x)$ ($\rho_a\ll 1$), leads to better convergence than the regular ReLU.

\begin{algorithm}[tb]
\SetAlgoLined
\footnotesize
\caption{Gradient descent based minimization}\label{alg:agd}
 Input: initial $\phi_i$, $\pi_i$, consts. $K_1, K_2, K_3, K_4$, $\epsilon_1$, $\epsilon_2$, $\epsilon_3$ ,$\epsilon_4$ \\ 
 \For{$i\leq \text{max steps}$}{
 $U(\pi, \phi)=J(\pi,\phi)+K_1 \rho(c_1-\epsilon_1) +K_2\rho(c_2-\epsilon_2) + K_3\rho(c_3-\epsilon_3) +K_4\rho(c_4-\epsilon_4)$\\
$\pi \leftarrow \pi - \alpha_{\pi}\frac{\nabla_\pi U(\pi,\phi)}{\|\nabla_\pi U(\pi,\phi)\|_2}\\
\phi \leftarrow \phi - 
\alpha_{\phi}\frac{\nabla_\phi
U(\pi,\phi)}{\|\nabla_\phi U(\pi,\phi)\|_2}$\\}
\Return $\; \pi, \phi$
\label{Algo}
\end{algorithm}

The method used for minimization is based on gradient descent with a normalized gradient since the raw gradient can be very large at certain points, and using this is likely to result in probabilities that are either negative or greater than one. $\alpha_{\phi}$ and $\alpha_{\pi}$ are the learning rates for $\phi$ and $\pi$, respectively.

\subsection{Threshold Initialization}
Since the problem is complicated non-convex, the initialization of $\pi$ and $\phi$ play a central role in obtaining a good policy. We initialize $\pi$ as a threshold policy, which has demonstrated to work well in practice. Specifically, we set $\pi(f,g) = 0$ for $fg < \tau$ and $\pi(f,g) = p$ for $fg \geq \tau$, where $\tau$ is a suitably chosen threshold. In our experiments we set $\tau$ equal to the mean AoII obtained using Algorithm~\ref{alg:agd} with initialization $\phi(f,g)\; \propto\; \frac{1}{fg}$, (normalized such that the sum is 1) and $\pi(f,g) \; \propto \;fg$, scaled such that $\pi \phi = 0.9$ transmission attempts, and $p=\frac{5}{N}$.
After empirically finding a good $\tau$ for a given value of $N$ and $p_t$, it can be extrapolated to other values of $p_t$ using a simple relation threshold $\tau\; \propto \; \sqrt{p_t}$, which can be shown to approximate the real AoII.

\section{Numerical Results}\label{sec:results}

In order to verify the quality of our optimization, we tested it in a Monte Carlo simulation over $10^5$ steps for each scenario. We compare our optimization with two state-independent benchmarks, which transmit regardless of the value of $f$ and $g$, as long as $g>0$, i.e., there is new information to send. The first strategy is to always transmit with probability $1/N$, and is dubbed PT1, while the second limits the load to $E$ by having sensors transmit with probability $E/N$.

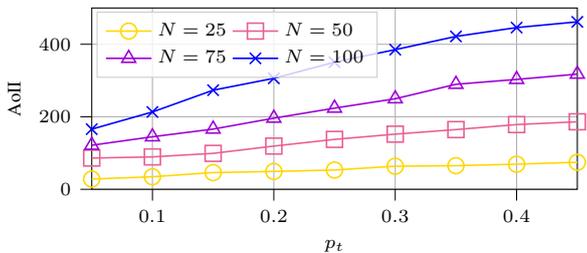
\begin{figure}
\begin{tikzpicture}

\definecolor{darkgray176}{RGB}{176,176,176}
\definecolor{lightgray204}{RGB}{204,204,204}

\begin{axis}[
width=\fwidth,
height=\fheight,
legend cell align={left},
legend style={
  fill opacity=0.8,
  legend columns=2,
  draw opacity=1,
  text opacity=1,
  at={(0.01,0.98)},
  anchor=north west,
  draw=lightgray204
},
font=\scriptsize,
tick align=outside,
tick pos=left,
x grid style={darkgray176},
xlabel={$p_t$},
xmajorgrids,
xmin=0.05, xmax=0.45,
xtick style={color=black},
y grid style={darkgray176},
ylabel={AoII},
ymajorgrids,
ymin=0, ymax=500,
ytick style={color=black}
]
\addplot [semithick, color3, mark=o, mark size=3, mark options={solid}]
table {%
0.05 28.26
0.1 34.57
0.15 46.12
0.2 49.25
0.25 53.18
0.3 63.65
0.35 65.08
0.4 69.22
0.45 74.56
};
\addlegendentry{$N=25$}
\addplot [semithick, color2, mark=square, mark size=3, mark options={solid}]
table {%
0.05 86.13
0.1 89.36
0.15 99.12
0.2 119.22
0.25 137.64
0.3 152.1
0.35 164.7
0.4 178.5
0.45 186.21
};
\addlegendentry{$N=50$}
\addplot [semithick, color1, mark=triangle, mark size=3, mark options={solid}]
table {%
0.05 121.49
0.1 144.87
0.15 166.08
0.2 196.17
0.25 223.84
0.3 249.72
0.35 289.61
0.4 303.06
0.45 317.62
};
\addlegendentry{$N=75$}
\addplot [semithick, color0, mark=x, mark size=3, mark options={solid}]
table {%
0.05 165.88
0.1 213.26
0.15 273.32
0.2 306.08
0.25 350.62
0.3 385.1
0.35 421.59
0.4 446.02
0.45 461.8
};
\addlegendentry{$N=100$}
\end{axis}

\end{tikzpicture}
\caption{Average AoII as a function of $p_t$.}
\label{fig:aoii}
\end{figure}

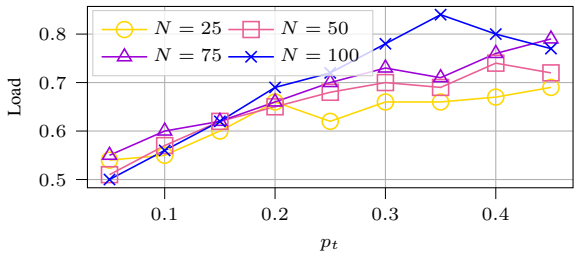
\begin{figure}
\begin{tikzpicture}

\definecolor{darkgray176}{RGB}{176,176,176}
\definecolor{lightgray204}{RGB}{204,204,204}

\begin{axis}[
width=\fwidth,
height=\fheight,
legend cell align={left},
legend style={
  fill opacity=0.8,
  draw opacity=1,
  text opacity=1,
  legend columns=2,
  at={(0.01,0.98)},
  anchor=north west,
  draw=lightgray204
},
  font=\scriptsize,
tick align=outside,
tick pos=left,
x grid style={darkgray176},
xlabel={$p_t$},
xmajorgrids,
xmin=0.03, xmax=0.47,
xtick style={color=black},
y grid style={darkgray176},
ylabel={Load},
ymajorgrids,
ymin=0.483, ymax=0.857,
ytick style={color=black}
]
\addplot [semithick, color3, mark=o, mark size=3, mark options={solid}]
table {%
0.05 0.54
0.1 0.55
0.15 0.6
0.2 0.66
0.25 0.62
0.3 0.66
0.35 0.66
0.4 0.67
0.45 0.69
};
\addlegendentry{$N=25$}
\addplot [semithick, color2, mark=square, mark size=3, mark options={solid}]
table {%
0.05 0.51
0.1 0.57
0.15 0.62
0.2 0.65
0.25 0.68
0.3 0.7
0.35 0.69
0.4 0.74
0.45 0.72
};
\addlegendentry{$N=50$}
\addplot [semithick, color1, mark=triangle, mark size=3, mark options={solid}]
table {%
0.05 0.55
0.1 0.6
0.15 0.62
0.2 0.66
0.25 0.7
0.3 0.73
0.35 0.71
0.4 0.76
0.45 0.79
};
\addlegendentry{$N=75$}
\addplot [semithick, color0, mark=x, mark size=3, mark options={solid}]
table {%
0.05 0.5
0.1 0.56
0.15 0.62
0.2 0.69
0.25 0.72
0.3 0.78
0.35 0.84
0.4 0.8
0.45 0.77
};
\addlegendentry{$N=100$}
\end{axis}

\end{tikzpicture}
\caption{Average load as a function of $p_t$ for the dual policy.}
\label{fig:energy}\vspace{-0.4cm}
\end{figure}

A leaky ReLU with a slope $\rho_a=10^{-6}$ was used in all computations. We used tolerance levels
$\epsilon_1=10^{-3}$, %
$\epsilon_2=10^{-6}$, $\epsilon_3=10^{-5}$, and $\epsilon_4=10^{-6}$, and penalties $K_1=10^{8}$, %
$K_2=10^{11}$, $K_3=10^{10}$, and $K_4=10^{11}$. In most of our experiments the value of $J(\pi,\phi)$ is in the order of $10^3$; thus, these values ensure that the product $K_iP_i$ is about two orders of magnitude greater than $J(\pi,\phi)$ whenever $c_i>\epsilon_i$.

We simulated the resulting policy for different values of $N$ and $p_t$, after initializing all sensors in state $(0,0)$. The time average AoII (the true AoII, not the truncated version) is shown in Fig.~\ref{fig:aoii}. As expected, the AoII increases with $p_t$ and $N$, and the growth seems to be approximately proportional to $\sqrt{p_t}$: as we will discuss later, this behavior is reflected by the chosen threshold. However, the effect of the number of sensors $N$ depends on $p_t$, as the total load on the network has a non-linear effect: as we need to maintain the slotted ALOHA system in its stability range to avoid a complete collapse, the time $f$ between subsequent transmissions can significantly increase, and so will the error $g$ as transitions in the chain accumulate. The AoII will be a product of these, so the effect compounds, making systems with faster transitions much more sensitive to an increased number of nodes. However, the steady state approximation becomes more accurate with $N$, so the dual policy gets closer to the optimum.

Fig.~\ref{fig:energy} shows the average total number of transmissions across all sensors per time slot obtained from these simulations. Interestingly, when transitions are rare, the total number of transmissions is almost independent from $N$: as the channel is less loaded, the priority is to avoid collisions, as sensors in larger networks are able to transmit less often and still obtain a good AoII performance. When $p_t>0.2$, transitions become frequent, and holding back transmissions enough to avoid congestion becomes suboptimal: in larger networks, sensors can achieve a better AoII by attempting to transmit and failing with a relatively high probability, rather than waiting and risking wasting some slots. This matches our earlier intuition of a non-linearity in the optimal behavior, that becomes more pronounced as the network approaches full load and becomes severely congested. Note that the number of transmissions is less than $1$ in all cases, never reaching higher than 0.85: as such, the optimal policy is also more energy-efficient than state-independent slotted ALOHA approaches. Note that we can come up with policies with lower energy consumption ($\approx 0.5-0.6$) for higher values of $p_t$. This can be done by adding an energy constraint with a penalty($p_e = \|\pi \phi-0.5\|_2^2\mathbb{I}(\pi \phi > 0.5))$  to the objective function. However, these policies result in a higher AoII (about double the AoII we obtain without the constraint).

\begin{figure}
\centering
\begin{tikzpicture}

\definecolor{darkgray176}{RGB}{176,176,176}
\definecolor{lightgray204}{RGB}{204,204,204}

\begin{axis}[
width=\sfwidth,
height=\sfheight,
legend cell align={left},
legend style={fill opacity=0.8, draw opacity=1, text opacity=1, draw=lightgray204,legend columns=2,
  at={(0.99,0.02)},
  anchor=south east,},
tick align=outside,
tick pos=left,
font=\scriptsize,
x grid style={darkgray176},
xlabel={$p_t$},
xmajorgrids,
xmin=0.03, xmax=0.47,
xtick style={color=black},
y grid style={darkgray176},
ylabel={AoII reduction (\%)},
ymajorgrids,
ymin=70, ymax=90,
ytick style={color=black}
]
\addplot [semithick, color3, mark=o, mark size=3, mark options={solid}]
table {%
0.05 74.8461538461538
0.1 79.0484848484848
0.15 79.8594175871468
0.2 79.9601236979167
0.25 78.1400778210117
0.3 79.5101351351351
0.35 78.8013029315961
0.4 79.0876132930514
0.45 78.7578347578348
};
\addlegendentry{$N=25$}
\addplot [semithick, color2, mark=square, mark size=3, mark options={solid}]
table {%
0.05 75.9413407821229
0.1 81.3833333333333
0.15 82.4566371681416
0.2 81.6584615384615
0.25 82.0313315926893
0.3 82.8135593220339
0.35 82.1366594360087
0.4 81.7484662576687
0.45 82.5482661668229
};
\addlegendentry{$N=50$}
\addplot [semithick, color1, mark=triangle, mark size=3, mark options={solid}]
table {%
0.05 80.9874804381847
0.1 82.7740784780024
0.15 82.2374331550802
0.2 83.059585492228
0.25 83.2329588014981
0.3 83.549383517197
0.35 82.9540906415539
0.4 82.9165727170237
0.45 83.1858125992589
};
\addlegendentry{$N=75$}
\addplot [semithick, color0, mark=x, mark size=3, mark options={solid}]
table {%
0.05 83.0040983606557
0.1 83.3650546021841
0.15 83.4652147610405
0.2 83.7967178401271
0.25 83.600561272217
0.3 83.6891147818721
0.35 83.2768742562475
0.4 83.8398550724638
0.45 84.0648723257419
};
\addlegendentry{$N=100$}
\end{axis}

\end{tikzpicture}
\caption{Performance improvement over the PT1 policy as a function of $p_t$.}
\label{fig:pt1}\vspace{-0.4cm}
\end{figure}
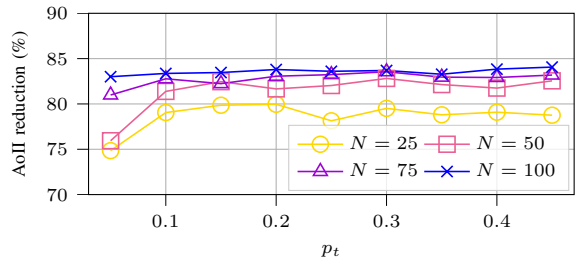

Next, we compare the dual policy with the PT1 and PTE benchmarks: Figs.~\ref{fig:pt1} and \ref{fig:pte} show that the performance improvement in terms of the average AoII is between 75\% and 85\%, i.e., the AoII achieved by the dual policy is about 5 times lower than for state-independent strategies. PT1 performs better than PTE, but while PTE has the same energy efficiency as our policy by design, the load generated by PT1 is always higher, causing a higher energy cost for sensors. The overall improvement is slightly higher for larger networks, but even for smaller networks the performance gain is significant.

\begin{figure}
\centering
\begin{tikzpicture}

\definecolor{crimson2143940}{RGB}{214,39,40}
\definecolor{darkgray176}{RGB}{176,176,176}
\definecolor{darkorange25512714}{RGB}{255,127,14}
\definecolor{forestgreen4416044}{RGB}{44,160,44}
\definecolor{lightgray204}{RGB}{204,204,204}
\definecolor{steelblue31119180}{RGB}{31,119,180}

\begin{axis}[
width=\sfwidth,
height=\sfheight,
legend cell align={left},
legend style={fill opacity=0.8, draw opacity=1, text opacity=1, draw=lightgray204,legend columns=2,
  at={(0.99,0.02)},
  anchor=south east,},
tick align=outside,
tick pos=left,
font=\scriptsize,
x grid style={darkgray176},
xlabel={$p_t$},
xmajorgrids,
xmin=0.03, xmax=0.47,
xtick style={color=black},
y grid style={darkgray176},
ylabel={AoII reduction (\%)},
ymajorgrids,
ymin=70, ymax=90,
ytick style={color=black}
]
\addplot [semithick, color3, mark=o, mark size=3, mark options={solid}]
table {%
0.05 81.2847682119205
0.1 84.2146118721461
0.15 83.0843373493976
0.2 81.1302681992337
0.25 81.5803278688525
0.3 80.987460815047
0.35 82.3152173913044
0.4 80.6648044692737
0.45 81.9903381642512
};
\addlegendentry{$N=25$}
\addplot [semithick, color2, mark=square, mark size=3, mark options={solid}]
table {%
0.05 81.3974082073434
0.1 85.5638126009693
0.15 85.7790530846485
0.2 84.3337713534823
0.25 84.1974741676234
0.3 84.4478527607362
0.35 84.8342541436464
0.4 83.8607594936709
0.45 84.4565943238731
};
\addlegendentry{$N=50$}
\addplot [semithick, color1, mark=triangle, mark size=3, mark options={solid}]
table {%
0.05 85.0565805658057
0.1 85.8801169590643
0.15 85.7686375321337
0.2 85.8565248738284
0.25 84.5627586206896
0.3 84.5278810408922
0.35 83.5821995464853
0.4 84.1163522012579
0.45 83.7948979591837
};
\addlegendentry{$N=75$}
\addplot [semithick, color0, mark=x, mark size=3, mark options={solid}]
table {%
0.05 87.3759512937595
0.1 86.7950464396285
0.15 85.6298633017876
0.2 86.6573670444638
0.25 85.523534269199
0.3 84.7424722662441
0.35 84.1686068343973
0.4 84.4157931516422
0.45 84.570664884731
};
\addlegendentry{$N=100$}
\end{axis}

\end{tikzpicture}
\caption{Performance improvement over the PTE policy as a function of $p_t$.}
\label{fig:pte}
\end{figure}
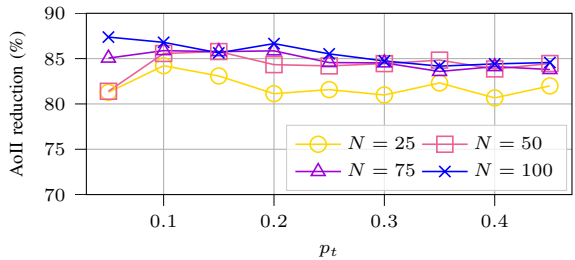

Finally, we can visualize the strategy itself: we take the case with $N=100$ and $p_t=0.3$ and plot the transmission probability as a colormap in Fig.~\ref{fig:p_visual}. We can see that the policy almost never transmits if the AoI is lower than 30, and that probability increases as $f$ and $g$ increases, but is very high if $f$ is saturated: in that case, the real AoII is unbounded, and transmission is relatively urgent. The maximum transmission rate in that case is close to 0.08, which is much higher than the normal rate. The case in which $g>f$, which is shown as having transmission probability 0 in the colormap, is never reached in practice, as we discussed above. We can also see the steady state probability for each state, mapped in Fig.~\ref{fig:ssd_visual}: in general, states with a relatively high age $f$ are reached often, but only if the error $g$ is very low. The unlucky case in which a sequence of transitions quickly leads the AoII to increase also has a relatively high probability, but as soon as the overall AoII passes the threshold, the transmission probability is correspondingly high.

\section{Conclusions and Future Work}\label{sec:conclusion}

In this work, we analyzed the optimization of the AoII in a slotted ALOHA network, in which sensors need to distributedly transmit updates about independent Markov processes. We consider a dual optimization based on a steady state approximation of other sensors to find a high-performance strategy to minimize AoII, starting from a threshold-based policy and gradually improving it.
We found that the policy outperforms naive approaches that do not take the sensor state into account and benchmark threshold policies, and that sensors can successfully coordinate to reduce AoII, even though the distributed scenario is much harder than a centralized one.

In future work, we plan to consider more advanced scenarios in which sensors' observations may be correlated, as well as dual approaches that combine polling and unprompted updates, taking the best from each approach to deal with complex environments which cannot be perfectly modeled.

\begin{figure}[t]
    \centering
    \includegraphics[width = 0.9\columnwidth]{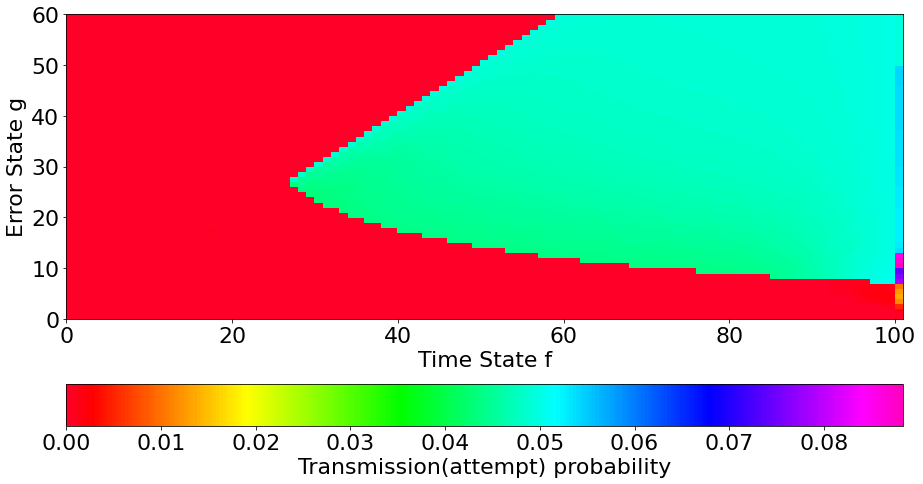}
    \caption{Visualization of the policy used ($p_t\!=\!0.3,N\!=\!100$).}
    \label{fig:p_visual}\vspace{-0.4cm}
\end{figure}

\begin{figure}[t]
    \centering
    \includegraphics[width = 0.9\columnwidth]{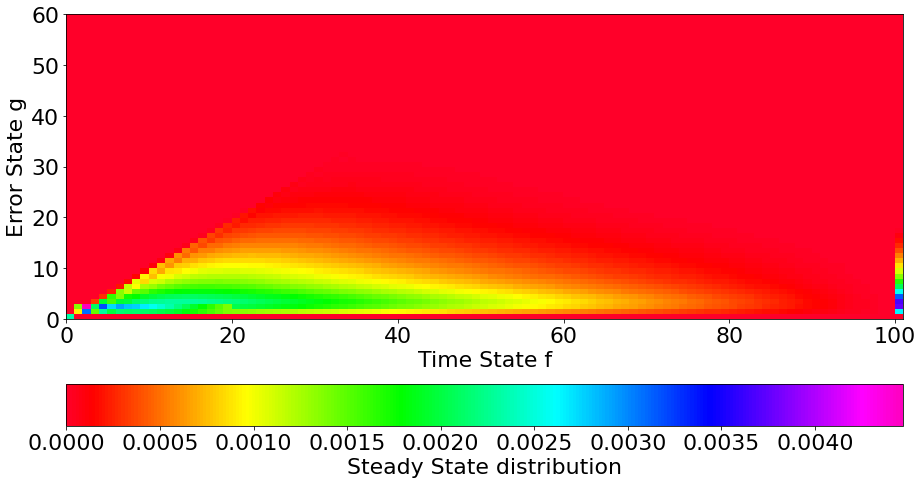}
    \caption{Steady-state distribution (optimized policy, $p_t\!=\!0.3,N\!=\!100$).}
    \label{fig:ssd_visual}\vspace{-0.4cm}
\end{figure}

\section*{Acknowledgment}
This work was supported by the Villum Investigator grant ``WATER'' from the Velux Foundation, Denmark. The work of A. E. Kal{\o}r was supported by the Independent Research Fund Denmark under Grant 1056-00006B. The work of F. Chiariotti was supported by the European Union under the Italian National Recovery and Resilience Plan (NRRP) of NextGenerationEU, under the REDIAL Young Researchers grant and the partnership PE0000001 - program
``RESTART''.

\bibliographystyle{IEEEtran}

\begin{thebibliography}{10}
\providecommand{\url}[1]{#1}
\csname url@samestyle\endcsname
\providecommand{\newblock}{\relax}
\providecommand{\bibinfo}[2]{#2}
\providecommand{\BIBentrySTDinterwordspacing}{\spaceskip=0pt\relax}
\providecommand{\BIBentryALTinterwordstretchfactor}{4}
\providecommand{\BIBentryALTinterwordspacing}{\spaceskip=\fontdimen2\font plus
\BIBentryALTinterwordstretchfactor\fontdimen3\font minus
  \fontdimen4\font\relax}
\providecommand{\BIBforeignlanguage}[2]{{%
\expandafter\ifx\csname l@#1\endcsname\relax
\typeout{** WARNING: IEEEtran.bst: No hyphenation pattern has been}%
\typeout{** loaded for the language `#1'. Using the pattern for}%
\typeout{** the default language instead.}%
\else
\language=\csname l@#1\endcsname
\fi
#2}}
\providecommand{\BIBdecl}{\relax}
\BIBdecl

\bibitem{kaul12}
S.~Kaul, R.~Yates, and M.~Gruteser, ``Real-time status: How often should one
  update?'' in \emph{Computer Communications Conference (INFOCOM)}.\hskip 1em
  plus 0.5em minus 0.4em\relax IEEE, 2012, pp. 2731--2735.

\bibitem{kosta17}
A.~Kosta, N.~Pappas, and V.~Angelakis, ``Age of {Information}: A new concept,
  metric, and tool,'' \emph{Foundations and Trends in Networking}, vol.~12,
  no.~3, pp. 162--259, 2017.

\bibitem{yates21}
R.~D. Yates, Y.~Sun, D.~R. Brown, S.~K. Kaul, E.~Modiano, and S.~Ulukus, ``Age
  of {Information}: An introduction and survey,'' \emph{IEEE Journal on
  Selected Areas in Communications}, vol.~39, no.~5, pp. 1183--1210, 2021.

\bibitem{ayan19}
O.~Ayan, M.~Vilgelm, M.~Kl{\"u}gel, S.~Hirche, and W.~Kellerer,
  ``Age-of-information vs. value-of-information scheduling for cellular
  networked control systems,'' in \emph{10th International Conference on
  Cyber-Physical Systems}.\hskip 1em plus 0.5em minus 0.4em\relax ACM/IEEE,
  2019, pp. 109--117.

\bibitem{10.1109/TNET.2020.3005549}
A.~Maatouk, S.~Kriouile, M.~Assaad, and A.~Ephremides, ``The {Age of Incorrect
  Information}: A new performance metric for status updates,'' \emph{IEEE/ACM
  Transactions on Networking}, vol.~28, no.~5, pp. 2215–--2228, 2020.

\bibitem{KSA21}
S.~Kriouile, M.~Assaad, and A.~Maatouk, ``On the global optimality of
  {Whittle}’s index policy for minimizing the {Age of Information},''
  \emph{IEEE Transactions on Information Theory}, vol.~68, no.~1, pp. 572--600,
  2022.

\bibitem{9518209}
S.~Kriouile and M.~Assaad, ``Minimizing the {Age of Incorrect Information} for
  real-time tracking of {Markov} remote sources,'' in \emph{International
  Symposium on Information Theory (ISIT)}.\hskip 1em plus 0.5em minus
  0.4em\relax IEEE, 2021, pp. 2978--2983.

\bibitem{Bjoshi}
B.~Joshi, R.~V. Bhat, B.~N. Bharath, and R.~Vaze, ``Minimization of age of
  incorrect estimates of autoregressive markov processes,'' in \emph{19th
  International Symposium on Modeling and Optimization in Mobile, Ad hoc, and
  Wireless Networks (WiOpt)}.\hskip 1em plus 0.5em minus 0.4em\relax IEEE,
  2021.

\bibitem{YA21}
Y.~Chen and A.~Ephremides, ``Minimizing {Age of Incorrect Information} for
  unreliable channel with power constraint,'' in \emph{Global Communications
  Conference (GLOBECOM)}.\hskip 1em plus 0.5em minus 0.4em\relax IEEE, 2021.

\bibitem{delay}
\BIBentryALTinterwordspacing
------, ``Preempting to minimize {Age of Incorrect Information} under random
  delay,'' 2022. [Online]. Available: \url{https://arxiv.org/abs/2209.14254}
\BIBentrySTDinterwordspacing

\bibitem{9162973}
H.~Chen, Y.~Gu, and S.-C. Liew, ``{Age-of-Information} dependent random access
  for massive {IoT} networks,'' in \emph{Conference on Computer Communications
  Workshops (INFOCOM WKSHPS)}.\hskip 1em plus 0.5em minus 0.4em\relax IEEE,
  2020, pp. 930--935.

\bibitem{vavascan21}
O.~T. Yavascan and E.~Uysal, ``Analysis of slotted aloha with an age
  threshold,'' \emph{IEEE Journal on Selected Areas in Communications},
  vol.~39, no.~5, pp. 1456--1470, 2021.

\bibitem{LD12}
L.~Dai, ``Stability and delay analysis of buffered aloha networks,'' \emph{IEEE
  Transactions on Wireless Communications}, vol.~11, no.~8, pp. 2707--2719,
  2012.

\bibitem{BJK05}
B.-J. Kwak, N.-O. Song, and L.~Miller, ``Performance analysis of exponential
  backoff,'' \emph{IEEE/ACM Transactions on Networking}, vol.~13, no.~2, pp.
  343--355, 2005.

\end{thebibliography}

\end{document}